%% file: Caching_final_AC_2nd_submission.tex
\documentclass[a4paper,conference]{IEEEtran}

\usepackage{amsfonts}
\usepackage{amsmath,amssymb}
\usepackage{hyperref}
\usepackage{amsthm}
\usepackage{float}
\usepackage{slashbox}
\usepackage{nicefrac}
\usepackage[dvipsnames]{xcolor}
\usepackage{enumitem}
\usepackage{tikz}
\usepackage{graphicx}   
\usepackage[makeroom]{cancel}
\usepackage{colortbl}
\usepackage{multirow,geometry,balance}
\theoremstyle{remark}
\newtheorem{theorem}{ {Theorem}}

\newtheorem{definition}{{Definition}}

\newtheorem{lemma}{ {Lemma}}
\newtheorem{remark}{ {Remark}}

\let\olddefinition\definition
\renewcommand{\definition}{\olddefinition\normalfont}
\let\oldtheorem\theorem
\renewcommand{\theorem}{\oldtheorem\normalfont}
\let\oldremark\remark
\renewcommand{\remark}{\oldremark\normalfont}
\usepackage{subcaption}
\usepackage{pgfplots}
\pgfplotsset{compat=newest}
\usetikzlibrary{plotmarks}
\usepackage{grffile}
\usepackage{amsmath}

\usetikzlibrary{shapes}
\pgfplotsset{compat=1.10}
 
\newcommand{\ti}{\textit}

\renewcommand{\IEEEQED}{\IEEEQEDopen}
\definecolor{OliveGreen}{rgb}{0,0.6,0}

\newcommand{\Mod}[1]{\ \mathrm{MOD}\ #1} 
\newcommand{\Modreg}[1]{\ \mathrm{mod}\ #1} 
\input{tikz_stuff}

\newcommand{\Alert}[1]{\textcolor{black}{#1}}

\begin{document}

\title{Delivery Time Minimization in Edge Caching: Synergistic Benefits of Subspace Alignment and Zero Forcing}

\author{\IEEEauthorblockN{Jaber Kakar$^{*}$, Alaa Alameer$^{*}$, Anas Chaaban$^{\dagger}$, Aydin Sezgin$^{*}$ and Arogyaswami Paulraj$^{\ddagger}$}
\IEEEauthorblockA{$^{*}$Institute of Digital Communication Systems,
Ruhr-Universit{\"a}t Bochum, Germany\\$^{\dagger}$Communication Theory Lab, King Abdullah University of Science and Technology, Thuwal, Saudi Arabia\\$^{\ddagger}$Information Systems Laboratory, Stanford University, CA, USA\\
Email: \{jaber.kakar, alaa.alameerahmad, aydin.sezgin\}@rub.de, anas.chaaban@kaust.edu.sa, apaulraj@stanford.edu
}}

\maketitle

\begin{abstract}
An emerging trend of next generation communication systems is to provide network edges with additional capabilities \Alert{such as} additional storage resources in \Alert{the} form of caches to reduce file delivery latency. 
To investigate this aspect, we study the fundamental limits of a cache-aided wireless network consisting of one central base station, $M$ transceivers and $K$ receivers from a latency-centric perspective. We use the normalized delivery time (NDT) to capture the per-bit latency for the worst-case file request pattern at high signal-to-noise ratios (SNR), \Alert{normalized} with respect to a reference interference-free system with unlimited transceiver cache capabilities. For various special cases with $M=\{1,2\}$ and $K=\{1,2,3\}$ that satisfy $M+K\leq 4$, we establish the optimal tradeoff between cache storage and latency. \Alert{This is facilitated} through establishing a novel converse (for arbitrary $M$ and $K$) and an achievability scheme on the NDT. Our achievability scheme is \Alert{a synergistic combination of} multicasting, zero-forcing beamforming and interference alignment. 
\end{abstract}


%
\IEEEpeerreviewmaketitle

\section{Introduction}
\label{sec:intro}

In the last decade, mobile usage in wireless networks has shifted from being connection-centric driven (e.g., phone calls) to content-centric (e.g., HD video) behaviors. In this context, integrating content caching in heterogeneous networks (HetNet) represents a viable solution for highly content-centric, next generation (5G) mobile networks. Specifically, when caching the most popular contents in  HetNet \emph{edge nodes}, e.g., eNBs and relays, alleviates backhaul traffic, reduces latency and ameliorates quality of service of mobile users. Thus, it is expected that future networks will be heterogeneous in nature, vastly deploying relay nodes (RN) (e.g., fixed RNs in LTE-A \cite{network_m2_2011} or mobile RNs in form of drones \cite{Kakar_Thesis,KakarUAV}) endowed with content cache capabilities. 

A simplistic HetNet modeling this aspect is shown in Fig. \ref{fig:HetNet}. In this model, $M$ RNs act as cache-aided transceivers. Thus, aspects of both transmitter and receiver caching in RNs is captured through this network model enabling a low \emph{delivery time} of requested files by $M$ RNs and $K$ user equipments (UE).\footnote{\Alert{We use the words \emph{delivery time} and \emph{latency} interchangeably}.} These terms refer to the timing overhead required to satisfy all file demands of requesting nodes in the network. In this work, we are interested in \emph{completely} characterizing the fundamental delivery time cache memory trade-off of this particular network for specific instances of $M$ and $K$.

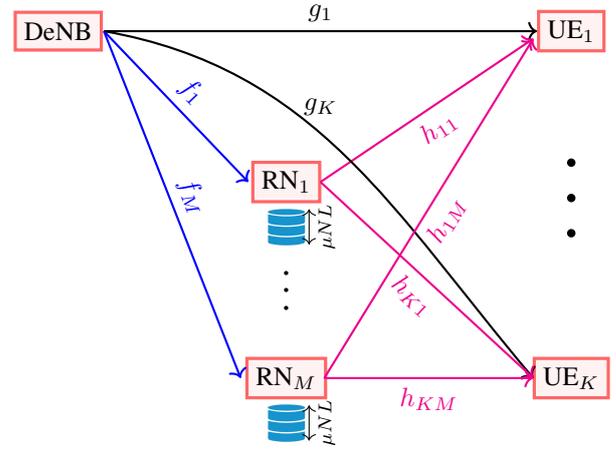
\begin{figure}[h]
	\vspace{3em}
	\begin{tikzpicture}[scale=0.8]
	\SymMod
	\end{tikzpicture}
	\vspace{-3.5em}
	\caption{\small A transceiver cache-aided HetNet consisting of one DeNB, $M$ RNs and $K$ UEs.}	
	\label{fig:HetNet}
\end{figure}

In prior work, it was shown that both receiver (Rx) and transmitter (Tx) caching can offer significant \Alert{latency reduction}. Rx caching was first studied in \cite{Maddah-Ali2} for a shared link with one server and multiple cache-enabled receivers. The authors show that \Alert{appropriate caching of} popular content facilitates multicast opportunities and consequently reduces latency. On the other hand, the impact of Tx caching on latency has mainly been investigated by analyzing the inverse degrees-of-freedom (DoF) metric of Gaussian interference networks \cite{Soheil}. To this end, the authors of \cite{Maddah_Ali} developed an interference alignment scheme characterizing the metric as a function of the cache storage size for a 3-user Gaussian interference network. The caches are prefetched to allow transmitter cooperation \Alert{so that} interference coordination techniques are applicable. The first lower bounds on the \Alert{inverse DoF} were developed in \cite{avik} for a network with an arbitrary number of edge nodes and users. With these bounds, the optimality of schemes presented in \cite{Maddah_Ali} for certain regimes of cache sizes \Alert{was} shown under uncoded prefetching of the cached content. Extensions of this work include the characterization of the latency-memory tradeoff \Alert{in} cloud and cache-assisted networks for \Alert{equally and non-equally strong wireless links in \cite{Tandon} and \cite{KakarArxiv,KakarICC}, respectively.} Recently, in two new lines of research, the effect of Tx-Rx caching at \emph{distinct} nodes \cite{Naderializadeh} and transceiver caching \cite{conference214} on the latency were investigated. This paper focuses on the latter. 

In this paper, we study the fundamental limits on the delivery time for a \emph{transceiver} cache-aided HetNet consisting of one donor eNBs (DeNB), $M$ transceivers and $K$ users. We measure the performance through \Alert{a} latency-centric metric \Alert{known as the} \emph{normalized delivery time per bit} (NDT) (cf. formal definition of NDT in Eq. \eqref{eq:NDT} in Section \ref{sec:Sym_Model}). This metric, first introduced in \cite{avik}, indicates the worst-case per-bit latency incurred in the wireless network with respect to a reference interference-free system without cache capacity restrictions. Similarly to the DoF, it is a high signal-to-noise ratio (SNR) metric. The main contributions of this paper are as follows: 
\begin{itemize}[leftmargin=*]
	\item We develop a novel class of information theoretic lower bounds on the NDT under the assumption of perfect channel state information (CSI) and uncoded prefetching of the cached content. 
	\item We completely characterize the NDT-cache memory tradeoff for the settings of (a) $M=1$ RNs and $K=\{1,2,3\}$ UEs and (b) $M=2$ RNs and $K=\{1,2\}$ UEs. To this end, we establish NDT-optimal schemes that synergistically design precoders facilitating zero-forcing (ZF) beamforming, multicasting and interference alignment. Our schemes are optimal for both time-variant and invariant channels requiring \emph{finite} signal dimensions (time, frequency, etc.). Further, we determine the optimal schemes for the extremal cases of no caching and full caching.
	\item Along with our results, we discuss the relationship between (sum) DoF and NDT. To this end, we assess the results from both a rate (e.g., DoF), and latency (e.g., NDT) perspective.  
\end{itemize}

\textbf{Notation:} For any two integers $a$ and $b$ with $a\leq b$, we define $[a:b]\triangleq\{a,a+1,\ldots,b\}$. When $a=1$, we simply write $[b]$ for $\{1,\ldots,b\}$. The superscript $(\cdot)^{\dagger}$ represents the transpose of a matrix. Furthermore, we define the function $(x)^{+}\triangleq\max\{0,x\}$ and the \emph{modified} modular operator $c=a\Mod\{b\}$ for integers $a$ and $b$ as $c=a$ if $a\leq b$ and $c=a\Modreg{b}$ if $a>b$.

\section{System Model}
\label{sec:Sym_Model}
We study the \emph{downlink} of a transceiver cache-aided HetNet as shown in Fig. \ref{fig:HetNet}. The network consists of $M$ causal full-duplex RNs and a donor eNB (DeNB) which serves $K$ UEs with its desired content over a shared wireless channel. Simultaneously, each RN also requests information from the DeNB. At every transmission interval, we assume that RNs and UEs request files from the set $\mathcal{W}$ of $N$ popular files, whose elements are all of $L$ bits in size. The transmission interval terminates when the requested files have been delivered. The system model, notation and main assumptions for a \emph{single} transmission interval are summarized as follows:
\begin{itemize}[leftmargin=*]
	\item Let $\mathcal{W}=\{W_1,\ldots,W_{N}\}$ denote the library of popular files, where each file $W_n$ is of size $L$ bits. Each file $W_n$ is chosen uniformly at random from $[2^{L}]$. UEs and RNs request files $W_{d_u}$, $\forall u\in[K]$, and $W_{d_r}$, $\forall r\in[K+1:M+K]$, from the library $\mathcal{W}$, respectively. The demand vector $\mathbf{d}\Alert{=(d_1,\ldots,d_{M+K})}\in[N]^{M+K}$ denotes the request pattern of RNs and UEs. 
	\item The RNs are endowed with a cache capable of storing $\mu NL$ bits, where $\mu\in[0,1]$ corresponds to the \emph{fractional cache size}. It denotes how much
	content can be stored at each RN relatively to the entire library $\mathcal{W}$. 
	\item The DeNB has access to all $N$ popular files of $\mathcal{W}$.  
	\item Global CSI at time instant $t$ is summarized by the channel vectors $\mathbf{f}[t]=\{f_{m}[t]\}_{m=1}^{m=M}\in\mathbb{C}^{M}$ and $\mathbf{g}[t]=\{g_{k}[t]\}_{k=1}^{k=K}\in\mathbb{C}^{K}$ and the channel matrix $\mathbf{H}[t]=\{h_{km}[t]\}_{k=1,m=1}^{k=K,m=M}\in\mathbb{C}^{K\times M}$. \Alert{Here,} $f_m$ and $g_k$ represent the complex channel coefficients from DeNB to RN$_m$ and UE$_k$, respectively, \Alert{while} $h_{km}$ is the channel from RN$_m$ to UE$_k$. We assume that all channel coefficients are assumed to be drawn i.i.d. from a continuous distribution.
\end{itemize}

Communication over the wireless channel occurs in two consecutive phases, (a) \emph{placement phase} followed by (b) \emph{delivery phase}. \Alert{These are detailed next, along with} the key performance metric termed as \emph{normalized delivery time per bit} (NDT). 
\vspace{.25em}
\paragraph{Placement phase} During this phase, every RN is given full access to the database of $N$ files. The cached content at RN$_m$ is generated through its individual caching function. 

	\begin{definition}(Caching function)\label{def_cache_fct} 
	\Alert{RN$_m$, $\forall m=1,\ldots,M$, maps each file $W_n\in\mathcal{W}$ to its local \emph{file cache content} 
		\begin{equation}
		S_{m,n}=\phi_{m,n}(W_{n}),\qquad\forall n=1,\ldots,N\nonumber.
		\end{equation} 
		All $S_{m,n}$ are concatenated to form the total cache content 
		\begin{equation}
		S_m=(S_{m,1},S_{m,2},\ldots,S_{m,N})\nonumber
		\end{equation} 
		at RN$_m$.}
	\end{definition}
\vspace{.25em}
Hereby, due to the assumption of symmetry in caching, the entropy $H(S_{m,n})$ of each component $S_{m,n}$, $n=1,\ldots,N$, is upper bounded by $\nicefrac{\mu NL}{N}=\mu L$. The definition of the caching function presumes that every file $W_i$ is subjected to individual caching functions. Thus, permissible caching policies allow for intra-file coding but avoid coding across files known as inter-file coding. Moreover, the caching policy is typically kept fixed over long transmission intervals. Thus, it is indifferent to the UEs request pattern and of channel realizations.     
\vspace{.25em}
\paragraph{Delivery phase} 
In this phase, a transmission policy at DeNB and all RNs is applied to satisfy the given requests $\mathbf{d}$ under the current channel realizations $\mathbf{f},\mathbf{g}$ and $\mathbf{H}$. Throughout the remaining definitions, we denote the number of channel uses required to satisfy all file demands by $T$.  
\vspace{.25em}
\begin{definition}(Encoding functions)\label{def_enc_fct} The DeNB encoding function at time instant \Alert{$t\in[T]$}
		\begin{equation}
		\psi_{s}^{[t]}:[2^{NL}]\times [N]^{M+K}\times\mathbb{C}^{Mt}\times\mathbb{C}^{Kt}\times\mathbb{C}^{Kt\times M}\rightarrow \mathbb{C}\nonumber
		\end{equation} 
		determines the DeNBs transmission signal $x_{s}[t]=\psi_{s}^{[t]}(\mathcal{W},\mathbf{d},\mathbf{f}_{t=1}^{t},\mathbf{g}_{t=1}^{t},\mathbf{H}_{t=1}^{t})$ subjected to an average power constraint of $P$. The encoding function of the causal full-duplex RN$_m$ at time instant $t\in[T]$ is defined by
		\begin{align}
		\psi_{r,m}^{[t]}:&[2^{\mu NL}]\times \mathbb{C}^{t-1}\hspace{-.1cm}\times [N]^{M+K}\hspace{-.1cm}\times\mathbb{C}^{Mt}\hspace{-.1cm}\times\mathbb{C}^{Kt}\hspace{-.1cm}\times\mathbb{C}^{Kt\times M}\rightarrow \mathbb{C},\nonumber
		\end{align} 
		\Alert{which determines the codeword} $x_{r,m}[t]=\psi_{r,m}^{[t]}(S_m,\mathbf{y}_{r,m}^{t-1},\mathbf{d},\mathbf{f}_{t=1}^{t},\mathbf{g}_{t=1}^{t},\mathbf{H}_{t=1}^{t})$ while satisfying the average power constraint given by the parameter $P$.
\end{definition}
\vspace{.25em}
Hereby, the \Alert{codewords $x_{s}[t]$ and $x_{r,m}[t]$ are transmitted over $t\in[T]$ channel uses.} For any time instant $t$, $\psi_{r,m}^{[t]}$ accounts for the simultaneous reception and transmission through incoming and outgoing wireless links at RN$_m$. To be specific, at the $t$--th channel use the encoding function $\psi_{r,m}^{[t]}$ maps the cached content $S_m$, the received signal $\mathbf{y}_{r,m}^{t-1}$ (see Eq. \eqref{eq:Gaus_mod_RN}), the demand vector $\mathbf{d}$ and global CSI to the codeword $x_{r,m}[t]$.

\Alert{After transmission, the received signals at UE$_k$ is given by
\begin{equation}\label{eq:Gaus_mod}
		y_{u,k}[t]=
		g_{k}[t]x_{s}[t]+\sum_{m=1}^{M}h_{km}[t]x_{r,m}[t]+z_{u,k}[t],\forall t\in[T],
		\end{equation}
		where $z_{u,k}[t]$ denotes complex i.i.d. Gaussian noise of zero mean and unit power. The received signal at RN$_{m}$ is given by
	\begin{equation}\label{eq:Gaus_mod_RN}
		y_{r,m}[t]=
		f_{m}[t]x_{s}[t]+z_{r,m}[t],\forall t\in[T],
		\end{equation}	
		where $z_{r,m}[t]$ is additive zero mean, unit-power i.i.d. Gaussian noise. The desired files are decoded using the following functions.}
\vspace{.25em}
	\begin{definition}(Decoding functions)\label{def_dec_fct} 
	\Alert{The decoding operation at UE$_k$ follows the mapping
		\begin{equation}
		\eta_{u,k}:\mathbb{C}^{T}\times [N]^{M+K}\times\mathbb{C}^{MT}\times\mathbb{C}^{KT}\times\mathbb{C}^{KT\times M}\rightarrow [2^{L}]\nonumber. 
		\end{equation} 
		to provide an estimate $\hat{W}_{d_k}=\eta_{u,k}(\mathbf{y}_{u,k}^{T},\mathbf{d},\mathbf{f}_{t=1}^{T},\mathbf{g}_{t=1}^{T},\mathbf{H}_{t=1}^{T})$ of the requested file $W_{d_k}$.
In contrast to decoding at UE$_k$, all RNs explicitly leverage their cached content according to
		\begin{align}
		\eta_{r,m}:\mathbb{C}^{T}\hspace{-.12cm}\times [2^{\mu NL}]\hspace{-.05cm}\times [N]^{M+K}\hspace{-.12cm}\times\mathbb{C}^{MT}\hspace{-.12cm}\times\mathbb{C}^{KT}\hspace{-.12cm}\times\mathbb{C}^{KT\times M}\hspace{-.15cm}\rightarrow [2^{L}]\nonumber 
		\end{align} 
		to generate $\hat{W}_{d_r}=\eta_{r,m}(\mathbf{y}_{r,m}^{T},S_m,\mathbf{d},\mathbf{f}_{t=1}^{T},\mathbf{g}_{t=1}^{T},\mathbf{H}_{t=1}^{T})$ as an estimate of the requested file $W_{d_r}$.}
\end{definition}

A proper choice of caching, encoding and decoding functions that satisfy the reliability condition; that is, the worst-case error probability 
\begin{equation}\label{eq:error_prob}
P_e=\max_{\mathbf{d}\in [N]^{M+K}}\max_{j\in[M+K]}\mathbb{P}(\hat{W}_{d_j}\neq W_{d_j})
\end{equation} 
approaches $0$ as $L\rightarrow\infty$, is called a \emph{feasible policy}. \Alert{Now we are ready to define the delivery time per bit and its normalized version.}
\vspace{.25em}
\begin{definition}(Delivery time per bit \cite{avik}) 
The delivery time per bit (DTB) for a given request pattern $\mathbf{d}$ and channel realization $\mathbf{f},\mathbf{g}$ and $\mathbf{H}$ is defined as 
	\begin{equation}\label{eq:DTB}
	\Delta(\mu,P)=\max_{\mathbf{d}\in [N]^{M+K}}\limsup_{L\rightarrow\infty}\frac{\mathbb{E}[T(\mathbf{d},\mathbf{f},\mathbf{g},\mathbf{H})]}{L},
	\end{equation} 
	where the expectation is over the channel realizations.
\end{definition}  
\vspace{.25em}
In \Alert{the definition above}, $T$ represents the completion or delivery time \cite{Liu2011}. The normalization of the expected delivery time by the file size $L$ gives insight about the per bit-latency. In this context, the DTB measures the per-bit latency, i.e., the latency incurred per-bit when transmitting the requested files through the wireless channel, within a single transmission interval for the \emph{worst-case} request pattern of RNs and UEs as $L\rightarrow\infty$. The DTB depends on the fractional cache size $\mu$ and the power level $P$.   

In analogy to the degrees-of-freedom metric, the normalized delivery time per bit (NDT) is a high-SNR metric that relates the DTB to that of a point-to-point reference system. 
\vspace{.25em}
\begin{definition} (Normalized delivery time \cite{avik}) 
The NDT is defined as 
\begin{equation}\label{eq:NDT}
\delta(\mu)=\lim_{P\rightarrow\infty}\frac{\Delta(\mu,P)}{1/\log(P)}.
	\end{equation} 
	The minimum NDT $\delta^{\star}(\mu)$ is the infimum in NDT of all \Alert{feasible policies}.
\end{definition}  
\vspace{.25em}
	The NDT \Alert{compares} the \emph{delivery time per bit} achieved by the feasible coding scheme for the worst-case demand scenario to that of a baseline interference-free system in the high SNR regime. The achievable scheme, on the one hand, allows for reliable transmission of one file of $L$ bits to a single Rx on average in $\mathbb{E}[T(\mathbf{f},\mathbf{g},\mathbf{H})]$ channel uses, i.e., $1$ bit in $\mathbb{E}[T(\mathbf{f},\mathbf{g},\mathbf{H})]/L$ channel uses. The baseline system (e.g., a point-to-point channel), on the other hand, can transmit $\log(P)$ bits to a single Rx in one channel use, i.e., $1$ bit in $1/\log(P)$ channel uses. Therefore, the resulting NDT $\delta(\mu)$ indicates that the worst-case delivery time for one bit of the cache-aided network at fractional cache size $\mu$ is $\delta(\mu)$ times larger that the time needed by the baseline system.   

From \cite[Lemma 1]{Maddah_Ali}, it readily follows that the NDT is a convex function in $\mu$. This means that a cache-aided network shown in Fig. \ref{fig:HetNet} operating at fractional cache size $\mu=\alpha\mu_1+(1-\alpha)\mu_2$ for any $\alpha\in [0,1]$ achieves at most an NDT equal to the \emph{convex combination} $\alpha\delta(\mu_1)+(1-\alpha)\delta(\mu_2)$ through applying known feasible schemes applicable at fractional cache sizes $\mu_1$ and $\mu_2$ on distinct $\alpha$ and $1-\alpha$-fractions of the files, respectively. This strategy is known as \emph{memory sharing}.

\section{Lower Bound (Converse) on NDT}
\label{sec:lw_bd}

For a given worst-case demand pattern $\mathbf{d}$; that is all $K$ UEs and $M$ RNs request \emph{distinct} files $W_{d_j}$ ($d_j\neq d_\ell,j\neq \ell$), and given channel realizations $\mathbf{f},\mathbf{g}$ and $\mathbf{H}$, we obtain a lower bound on the delivery time $T=T(\mathbf{d},\mathbf{f},\mathbf{g},\mathbf{H})$, and therefore ultimately on the NDT, of any \emph{feasible} scheme. Note that $K+M$ distinct files $W_{d_k}$ are available if there are at least as many files in the library, i.e., $N\geq K+M$. Without loss of generality, we assume that the requested files by the $K$ UEs are $W_{[1:K]}=(W_1,W_2,\ldots,W_K)$ and of the $M$ RNs $W_{[K+1:K+M]}=(W_{K+1},W_{K+2},\ldots,W_{K+M})$. 

The key idea in establishing the lower bound on the NDT is that $K+\ell$ requested files, comprising of all $K$ files $W_{[1:K]}$ requested by the UEs and $\ell$ files desired by a subset of $\ell$ RNs (out of $M$ RNs), e.g., $W_{[K+1:K+\ell]}$, can be retrieved in the high SNR regime from 

\begin{itemize}[leftmargin=*]
	\item $s$ output signals of the UEs, e.g., $\mathbf{y}^{T}_{u,[1:s]}$ for $1\leq s\leq \min\{M+1,K\}$, and
	\item $\ell$ cached contents of $\ell$ RNs, e.g., $S_{[1:\ell]}$, where $\bar{s}\leq\ell\leq M$ and $\bar{s}=M+1-s$. 
\end{itemize} 

We note that since $s+\ell\geq M+1$ holds, we are able to reconstruct all $M+1$ transmit signals ($x_s[t]$ and $x_{r,m}[t],m\in[M]$) at all $T$ time instants of the delivery phase within bounded noise. The intuition behind the bound follows from \cite{avik,conference214}. Applying standard information-theoretic bounding techniques, results in the following Lemma.   
\vspace{.25em}
\begin{lemma}\label{lemma_lower_bound}
For the transceiver cache-aided network with one DeNB, $M$ RNs each endowed with a cache of fractional cache size $\mu\in[0,1]$, $K$ UEs and a file library of $N\geq M+K$ files, the NDT is lower bounded under perfect CSI at all nodes by
\begin{align}\label{eq:NDT_lw_bound}
\delta^{\star}&(\mu)\geq\max\Big\{1,\max_{\substack{\ell\in[\bar{s}:M],\\s\in[\min\{M+1,K\}]}}\delta_{\text{LB}}(\mu,\ell,s)\Big\},
\end{align} where $\bar{s}=M+1-s$ and 
\begin{align}\label{eq:NDT_lw_bound_inner_comp}
&\hspace{-.25cm}\delta_{\text{LB}}(\mu,\ell,s)=\frac{K+\ell-\mu(\bar{s}\big(K-s+\frac{(\bar{s}-1)}{2}\big)+\frac{\ell}{2}(\ell+1))}{s}.
\end{align}
\end{lemma}
\begin{proof}
The key idea behind the proof is provided in the previous paragraph. Details are omitted for the sake of brevity.
\end{proof}

\section{Achievability for some Special Cases}
\label{sec:ach_scheme}

\begin{figure*}
	\centering
	\begin{subfigure}[b]{0.475\textwidth}
		\centering
		\begin{tikzpicture}[scale=0.7]
		\PlotNDTMOneKOne
		\end{tikzpicture}
		\caption[NDT for $M=K=1$]{\small NDT for $M=K=1$}
		\label{fig:NDT_M_K_1_1}
	\end{subfigure}
	\hfill
	\begin{subfigure}[b]{0.475\textwidth}  
		\centering 
		\begin{tikzpicture}[scale=0.7]
		\PlotNDTMOneKTwo
		\end{tikzpicture}
		\caption[NDT for $M=1,K=2$]{\small NDT for $M=1, K=2$} 
		\label{fig:NDT_M_K_1_2}       
		\label{fig:th_2}
	\end{subfigure}
	\caption[NDT as a function of $\mu$ for $M=1$ and $K\leq 2$]
	{\small Optimal NDT as a function of $\mu$ for $M=1$ and $K\leq 2$.} 
	\label{fig:NDT_M_K_1_LEQ_2}
\end{figure*}
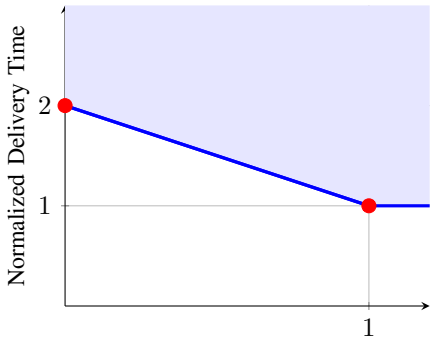
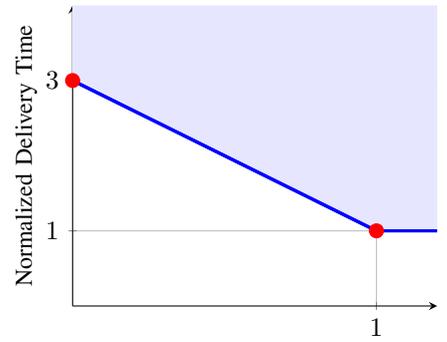 
 
First let us consider two special corner points at fractional cache sizes $\mu=0$ and $\mu=1$. These are the cases where the RN has either \emph{zero-cache} ($\mu=0$) or \emph{full-cache} ($\mu=1$) capabilities. We now expound the optimal NDT for these two points. 
\vspace{.5em}
\begin{lemma}\label{corr_mu_0_and_1}
For the transceiver cache-aided network with one DeNB, $M$ RNs each endowed with a cache of fractional cache size $\mu$, $K$ UEs and a file library of $N\geq M+K$ files, the optimal NDT is
\begin{equation}\label{eq:opt_NDT_mu_0}
\delta^{\star}(\mu)=K+M\:\:\text{ for }\mu=0,
\end{equation} 
achievable via DeNB broadcasting to $M$ RNs and $K$ UEs, and 
\begin{equation}
\label{eq:opt_NDT_mu_1}
\delta^{\star}(\mu)=\max\Bigg\{\frac{K}{M+1},1\Bigg\}\:\:\text{ for }\mu=1,
\end{equation} 
achievable via zero-forcing beamforming for an $(M+1,K)$ MISO\footnote{In MISO broadcast channels, we use the notation, $(a,b)$ for integers $a$ and $b$ to denote a broadcast channel with $a$ transmit antennas and $b$ single antenna receivers.} broadcast channel. 
\end{lemma}
\begin{proof}
	For the proof, it suffices to find a cache transmission policy that matches the lower bound in Lemma \ref{lemma_lower_bound} for $\mu=0$ and $\mu=1$, respectively. On the one hand, if $\mu=0$, we note that $\delta_{\text{LB}}(0,M,1)=K+M$. On the other hand, if $\mu=1$, we observe that $\delta_{\text{LB}}(1,0,M+1)=\nicefrac{K}{(M+1)}$ if $M+1\leq K$ and $\delta_{\text{LB}}(1,\ell,s)<1$ if $M+1>K$. Next, we consider the achievability at $\mu=0$ and $\mu=1$. For these two fractional cache sizes, the network in Fig. \ref{fig:HetNet} reduces to a SISO broadcast channel (BC) with $K+M$ users for $\mu=0$ and an $(M+1,K)$ MISO broadcast channel \Alert{for $\mu=1$}. The approximate \emph{per-user} rate (neglecting $o(\log(P))$ bits) for these two channels are known to be $\frac{1}{(K+M)}\log(P)$ (achievable through unicasting each user's message) and $\frac{1}{K}\min\{M+1,K\}\log(P)$ (achievable through zero-forcing beamfoming), respectively. Equivalently, each user needs the reciprocal per-user rate of signaling dimensions (e.g., channel uses in time or frequency) to retrieve one desired bit. Thus, the approximate DTB becomes, respectively, $\frac{(K+M)}{\log(P)}$ and $\frac{K}{\min\{M+1,K\}\log(P)}$. Normalizing the delivery time per bit by the \Alert{point-to-point reference} DTB $\frac{1}{\log(P)}$ generates the NDTs $K+M$ and $\max\{\nicefrac{K}{(M+1)},1\}$. This establishes the NDT-optimality at these fractional cache sizes.
\end{proof}
\vspace{.5em}
\begin{remark}
	From Lemma \ref{corr_mu_0_and_1}, we infer that the caching problem for the system illustrated in Fig. \ref{fig:HetNet} establishes the behavior of the network in terms of delivery time between the two extremes -- SISO BC with $K+M$ users and an $(M+1,K)$ MISO BC. This analysis will reveal what kind of schemes other than simple unicasting and zero-forcing will be optimal for \Alert{$0<\mu<1$.}
\end{remark}
\vspace{.5em}
Now, we move to special cases of the system where $M=\{1,2\}$ to provide a \emph{complete} characterization of the NDT-memory trade-off. We will primarily focus on $M=1$ for the sake of brevity.

\subsection{Achievability for $M=1$}
The lower bounds on the NDT are obtained from Lemma \ref{lemma_lower_bound} by setting $M=1$, yielding
\begin{align}
\label{eq:lower_bound_M_1}
\delta^{\star}(\mu)\geq\begin{cases}
\delta_{\text{LB}}(\mu,1,1)=K+1-\mu K&\text{ for } K\geq 1\\
\delta_{\text{LB}}(\mu,1,2)=\frac{K+1-\mu}{2}&\text{ for }K\geq 2\\\delta_{\text{LB}}(\mu,0,2)=\frac{K}{2}&\text{ for }K\geq 2
\end{cases}.
\end{align} 

For $K\leq 2$, the optimal NDT-memory curves are shown in Fig. \ref{fig:NDT_M_K_1_LEQ_2} for both $K=1$ (cf. Fig. \ref{fig:NDT_M_K_1_1}) and $K=2$ (cf. Fig. \ref{fig:NDT_M_K_1_2}). The achievability at the corner points (marked by 
circles in Fig. \ref{fig:NDT_M_K_1_LEQ_2}) at $\mu=0$ and $\mu=1$ readily follow from Lemma \ref{corr_mu_0_and_1}. Intermediary points are achievable through memory sharing. Thus, the optimal NDT for $K\leq 2$ and $M=1$ becomes 
\begin{equation}
\delta^{\star}(\mu)=K+1-\mu K.
\end{equation} This result is in agreement with our prior work \cite{conference214}. For $K>2$, on the other hand, the lower bound in \eqref{eq:lower_bound_M_1} simplifies to
\begin{align}\label{eq:lower_bound_M_1_K_great_2}
\delta^{\star}(\mu)\geq\begin{cases}
K+1-\mu K&\text{ for }0\leq\mu\leq\mu_1\\
\frac{K+1-\mu}{2}&\text{ for } \mu_1\leq\mu\leq 1
\end{cases},
\end{align} where $\mu_1=\frac{K+1}{2K-1}$ as shown in Fig. \ref{fig:NDT_M_K_1_3}. In order to show the tightness of the lower bound, we have to focus on the corner point $(\frac{K+1}{2K-1},\frac{K^{2}-1}{2K-1})$. Interestingly, if this point is achievable we make two observations.  

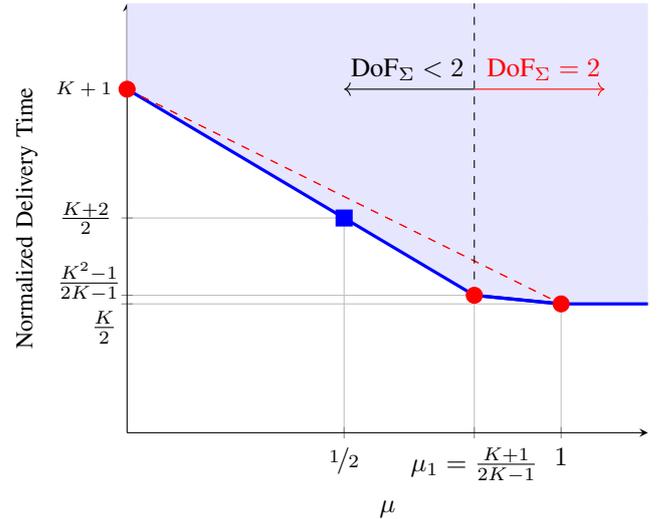
\begin{figure}
	\centering
	\begin{tikzpicture}[scale=1]
	\PlotNDTMOneKGreatTwo
	\end{tikzpicture}
	\caption[NDT lower bound for $M=1$ and $K>2$]{\small NDT lower bound for $M=1$ and $K\geq 3$. For $K=3$, this line is in fact achievable. The dashed line shows the achievable NDT of a \emph{suboptimal} time-sharing based unicasting-zero-forcing scheme, which is optimal for $M=1$ and $K\leq 2$.} 
	\label{fig:NDT_M_K_1_3}
\end{figure}
\vspace{.5em}
\begin{remark}
For increasing $K$, the point converges to $(\frac{1}{2},\frac{K}{2})$. This fact shows that \Alert{as $K$ increases}, a fractional memory size of $\mu_1\rightarrow\nicefrac{1}{2}$ suffices in already attaining the lowest attainable NDT of $\nicefrac{K}{2}$. 
The point $(\frac{1}{2},\frac{K+2}{2})$ (marked by a square in Fig. \ref{fig:NDT_M_K_1_3}) is achievable by leveraging interference alignment techniques for a $2\times K$ X-channel \cite{Cadambe09} and unicasting uncached information about the RNs desired file from the DeNB. 
\end{remark}
\vspace{.5em}
\begin{remark}
\Alert{If \eqref{eq:lower_bound_M_1_K_great_2} is achievable,} one can see that the per-user DoF of the $K$ UEs is $\frac{2K-1}{K^{2}-1}$ and that of the RN $\frac{K-2}{K^{2}-1}$\footnote{This is due to the fact that $L(1-\mu_1)$ symbols are uncached and are conveyed in $T=K^2-1$ channel uses, with $L=2K-1$. This constitutes the aforementioned DoF value.}. The resulting sum DoF thus becomes $\frac{K(2K-1)}{K^{2}-1}+\frac{K-2}{K^{2}-1}=2.$ Thus, at fractional cache sizes greater than $\mu_1$, the sum DoF remains $2$. This is shown in Fig. \ref{fig:NDT_M_K_1_3}.   
\end{remark}
\vspace{.5em}
So far, we were able to establish the achievability for this corner point for $K=3$ UEs only. The \emph{generalization} to arbitrary numbers of UEs is still an \emph{open problem}. In the sequel, we will illustrate the achievability of the corner point \Alert{$(\frac{K+1}{2K-1},\frac{K^{2}-1}{2K-1})=(\frac{4}{5},\frac{8}{5})$ for $K=3$.}

\begin{figure} 
	\FilesKThree
	\caption[Requested files by $K=3$ users and $M=1$ RNs]{\small Requested files by $K=3$ users and $M=1$ RNs and \Alert{the availability illustrated by the symbols transmitted from} the DeNB only or from both at the DeNB and the RN.}
	\label{fig:files_K_three}
\end{figure}

Assume without loss of generality $N=4$ and that the UEs request files $W_1,W_2$ and $W_3$ while the RN is interested in file $W_4$. 
According to Fig. \ref{fig:files_K_three}, all files are comprised of $5$ symbols, i.e., the $i$-th file is composed of symbols $\eta_{i,1},\eta_{i,2},\eta_{i,3},\eta_{i,4}$ and $\eta_{i,5}$. 
All these symbols are available at the DeNB. However, as far as the RN is concerned, only the first four symbols of all files are locally available in its cache. 
Since, the RN is interested in file $W_4$ and it knows $\eta_{4,1},\eta_{4,2},\eta_{4,3}$ and $\eta_{4,4}$, the only missing symbol it desires is $\eta_{4,5}$. 
\begin{figure} 
	\ZFKThree
	\caption[Zero-forcing map]{\small Map that assigns which symbol to zero-force at which receiver.} 
	\label{fig:ZF_K_three}
\end{figure} 
Thus, the transmission policy has to be designed such that DeNB and RN are involved in sending \emph{all} symbols of files $W_1,W_2$ and $W_3$ as well as $\eta_{4,5}$. These are in total $16$ information symbols. 
The transmission strategy will exploit the correlation that arises between the availability of shared symbols at RN and DeNB by leveraging zero-forcing (ZF) opportunities while \emph{simultaneously} facilitating (subspace) interference alignment (IA) at the UEs. This is why our scheme (as shown in Fig. \ref{fig:files_K_three}) only zero-forces symbols $\eta_{1,1},\eta_{1,2},\eta_{1,3}$, $\eta_{2,1},\eta_{2,2},\eta_{2,3}$ and $\eta_{3,1},\eta_{3,2},\eta_{3,3}$. Symbols $\eta_{1,4},\eta_{2,4}$ and $\eta_{3,4}$ are \emph{not} zero-forced but are instead used to enable alignment\footnote{IA is facilitated by the fact that the DeNB does \emph{not} transmit these symbols (even though it knows them). Thus, effectively, the DeNB does not need to be aware of $\eta_{i,4},i\in[N]$.} amongst others with $\eta_{4,5}$ at the UEs. The map that assigns which symbol is zero-forced at which UE is given in Figure \ref{fig:ZF_K_three}. To this end, DeNB and RN form their transmit signals according to 
\begin{align}
\label{eq:tx_sig_DeNB}
x_s[t]&=\sum_{i=1}^{3}\sum_{j=1,j\neq 4}^{5}\nu_{\eta_{i,j}}[t]\eta_{i,j}+\nu_{\eta_{4,5}}[t]\eta_{4,5},\\
\label{eq:tx_sig_RN}
x_r[t]&=\sum_{i=1}^{3}\sum_{j=1}^{4}\beta_{\eta_{i,j}}[t]\eta_{i,j},
\end{align} $\forall t\in[T]$ for $T=8$, respectively. The complex scalars $\nu_{\eta_{i,j}}[t]$ and $\beta_{\eta_{i,j}}[t]$ are precoders for symbol $\eta_{i,j}$ originating from DeNB and RN at time instant $t$, respectively. They are chosen such that both ZF and IA at the UEs become feasible. 
According to the ZF map of Fig. \ref{fig:ZF_K_three}, the ZF conditions at UE$_k$, $k\in[K]=[3]$, become
\begin{subequations}\label{eq:ZF_conditions}
	\begin{align}
	&\nu_{\eta_{(k+1)\Mod{K},1}}[t]g_k[t]+\beta_{\eta_{(k+1)\Mod{K},1}}[t]h_{k1}[t]=0,\\
	&\nu_{\eta_{(k+1)\Mod{K},2}}[t]g_k[t]+\beta_{\eta_{(k+1)\Mod{K},2}}[t]h_{k1}[t]=0,\\
	&\nu_{\eta_{(k+2)\Mod{K},3}}[t]g_k[t]+\beta_{\eta_{(k+2)\Mod{K},3}}[t]h_{k1}[t]=0.
	\end{align}
\end{subequations}

\begin{figure*}
	\centering
	\begin{tikzpicture}[scale=0.8]
	\AlignGraph
	\end{tikzpicture}
	\caption[Interference Alignment Graph for the achievability at corner point $(\frac{4}{5},\frac{8}{5})$ for $M=1$ and $K=3$]{\small Interference Alignment Graph for the achievability at corner point $(\frac{4}{5},\frac{8}{5})$ for $M=1$ and $K=3$. The graph consists of three (subspace) alignment chains.} 
	\label{fig:IA_Graph}
\end{figure*}
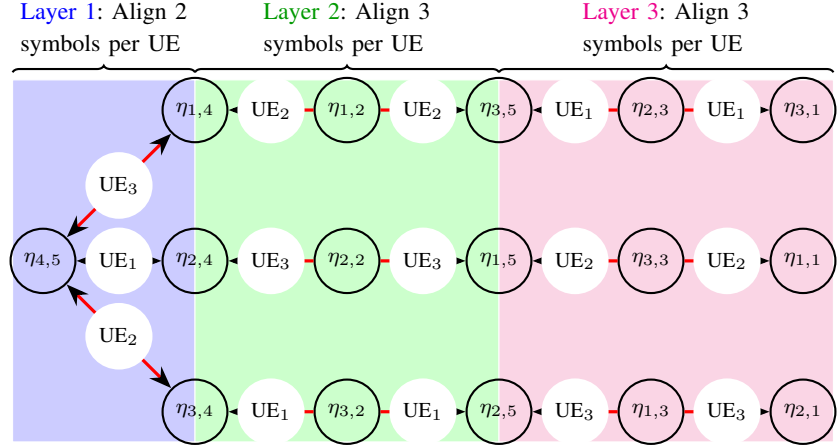

Simultaneously, we design the precoding scalars such that the interference at each UE is aligned into a three-dimensional signal space. (The remaining $5$ dimensions are reserved for the $5$ symbols of the desired file.) The interference graph in Fig. \ref{fig:IA_Graph} shows which symbols align with each other at which UE. This graph consists of $3$ layers. In the first layer, two symbols, namely $\eta_{4,5}$ and $\eta_{1,4}$, $\eta_{2,4}$ or $\eta_{3,4}$ align at the three UEs. At layers two and three, on the other hand, three symbols align per UE. Symbols $\eta_{1,4},\eta_{2,4}$ and $\eta_{3,4}$ link layers $1$ and $2$, while $\eta_{1,5},\eta_{2,5}$ and $\eta_{3,5}$ connect layers $2$ and $3$. In analogy to the graph in Fig. \ref{fig:IA_Graph}, the alignment conditions at UE$_k$ can be written \Alert{as 
\begin{align}
\label{eq:IA_condition1}
\nu_{\eta_{4,5}}[t]g_k[t]&=\beta_{\eta_{(k+1)\Mod{K},4}}[t]h_{k1}[t]
\end{align}
for Layer 1, 
\begin{align}
\label{eq:IA_condition2}
&\beta_{\eta_{(k+2)\Mod{K},4}}[t]h_{k1}[t]\nonumber\\
&=\beta_{\eta_{(k+2)\Mod{K},2}}[t]h_{k1}[t]+\nu_{\eta_{(k+2)\Mod{K},2}}[t]g_{k}[t]\nonumber\\
&=\nu_{\eta_{(k+1)\Mod{K},5}}[t]g_{k}[t]
\end{align}
for Layer 2, and 
\begin{align}
\label{eq:IA_condition3}
&\nu_{\eta_{(k+2)\Mod{K},5}}[t]g_{k}[t]\nonumber\\
&=\beta_{\eta_{(k+2)\Mod{K},1}}[t]h_{k1}[t]+\nu_{\eta_{(k+2)\Mod{K},1}}[t]g_{k}[t]\nonumber\\
&=\beta_{\eta_{(k+1)\Mod{K},3}}[t]h_{k1}[t]+\nu_{\eta_{(k+1)\Mod{K},3}}[t]g_{k}[t]
\end{align}    
for Layer 3.} Under the given ZF and IA conditions (cf.  \eqref{eq:ZF_conditions} and \eqref{eq:IA_condition1}--\eqref{eq:IA_condition3}), the precoders are functions of the channels $\mathbf{g}[t]$ and $\mathbf{H}[t]$. We fix the precoder for symbol $\eta_{4,5}$ to
\begin{equation}\label{eq:precoder_d5}
\nu_{\eta_{4,5}}[t]=j_{13}[t]j_{23}[t]j_{33}[t]g_1[t]g_2[t]g_3[t]h_{11}[t]h_{21}[t]h_{31}[t],
\end{equation} where
\begin{subequations}
	\begin{align}
	j_{13}[t]&=g_2[t]h_{31}[t]-g_{3}[t]h_{21}[t],\\
	j_{23}[t]&=g_3[t]h_{11}[t]-g_{1}[t]h_{31}[t],\\
	j_{33}[t]&=g_1[t]h_{21}[t]-g_{2}[t]h_{11}[t].
	\end{align}		
\end{subequations} 

We omit the solutions of the remaining precoders for the sake of brevity. However, these scalars can be computed by using \eqref{eq:precoder_d5} in \eqref{eq:ZF_conditions} and  \eqref{eq:IA_condition1}--\eqref{eq:IA_condition3}. Note that our approach also works with constant channels under the umbrella of \emph{real interference alignment} \cite{Motahari14,Maddah-Ali10}. Whether a two-phase precoding design with constant channels (similar to previous work on relay-aided X-channels \cite{Frank2014}) attains close-to-optimal performance is an interesting extension to work on. However, it is beyond the scope of this paper. Now we will go through the decoding from the perspective of both the RN and the UEs. 


At the receiver side, each UE observes interference aligned to $3$ independent signal dimensions and $5$ desired symbols occupying $5$ independent dimensions. In total, we require $T=8$ channel uses to allow for reliable decoding. Thus, the achievable DoF of each UE becomes $\nicefrac{5}{8}$ and that of the RN (it only requires $\eta_{4,5}$) $\nicefrac{1}{8}$. The NDT, on the other hand, corresponds to $\nicefrac{8}{5}$. This establishes the achievability for $M=1$ and $K=3$.

\subsection{Achievability for $M=2$}    

As we increase the number of RNs from $M=1$ to $M=2$, the concept of \emph{cooperative interference neutralization} becomes relevant. This enables the exploitation of side information at the RNs to allow them to receive their desired symbols while zero-forcing their contribution at the UEs. With this approach, we are able to show the achievability (and as such the optimality) for all corner points when $M=2$ and $K=\{1,2\}$. Details on the achievability is left out due to page limitations. 
In short, the optimal NDTs, are, respectively, given by
\begin{subequations}
	\begin{align}
	\delta^{\star}(\mu)&=\max\big\{3-4\mu,1\big\}\text{ for }M=2,K=1,\\
		\delta^{\star}(\mu)&=\max\bigg\{4-6\mu,\frac{4-3\mu}{2},\frac{3-\mu}{2}\bigg\}\text{ for }M=2,K=2.
	\end{align}
\end{subequations}

%

\section{Conclusion}
\label{sec:conlusion}

In this paper, we studied the fundamental limits on the delivery time for cache-aided wireless networks where relay nodes act as cache-equipped transceivers. We utilized the normalized delivery time (NDT) as a delivery time metric which captures the worst-case latency of the requested file retrieval. To this end, we developed, on the one hand, a novel lower bound for a cache-aided network with $M$ relay nodes and $K$ users. On the other hand, we determined NDT-optimal schemes with which we were able to completely characterize the trade-off between delivery time and cache memory for specific instances where $M=\{1,2\}$ and $K=\{1,2,3\}$ and $M+K\leq 4$. Our achievability schemes determine optimal precoders such that zero-forcing, interference alignment and cooperative interference neutralization are synergistically combined. The presented schemes are applicable to both time-variant and time-invariant channels. In future work, we would like to generalize our scheme. Specifically, we first aim to determine whether our scheme can be generalized for $M=1$ and $K>3$.    
    


\bibliographystyle{IEEEtran}
\bibliography{bibliography}
\balance
 
\end{document}

%% file: tikz_stuff.tex
\usetikzlibrary{decorations.text}
\usepgfplotslibrary{fillbetween}
\usetikzlibrary{calc, fit}
\usetikzlibrary{arrows.meta}
\usetikzlibrary{patterns}

\newcommand{\SymMod}[0]{
	\begin{tikzpicture}[
	roundnode/.style={circle, draw=green!60, fill=green!5, very thick, minimum size=7mm},
	squarednode/.style={rectangle, draw=red!60, fill=red!5, very thick, minimum size=5mm},
	]
	\node[squarednode] (eNB) at (0.0, 2.6+2) {DeNB};
	\node[squarednode] (HeNB) at (2+1, 0.6+2) {RN$_1$};
	\node[squarednode] (Utwo) at (3.75+3, 2.6+2) {UE$_1$};
	
	\draw[->, black, thick] (0.6, 2.6+2) -- (3.3+3, 2.6+2) node[pos=0.5,sloped,above] {$g_{1}$};
	\draw[->, thick, blue] (0.6, 2.6+2) -- (1.5+1, 0.6+2) node[pos=0.5,sloped,above] {$f_1$};
	\draw[->, magenta, thick] (2.45+1, 0.6+2) -- (3.25+3, 2.5+2) node[pos=0.5,sloped,below] {$h_{11}$};
	
	\node (ch) at (2+1, 0+2) {\includegraphics[scale=.3]{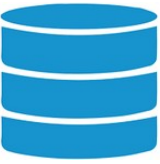}};
	\draw[<->] (2.3+1, -0.23+2) -- (2.3+1, 0.24+2) node[pos=0.5,sloped, below] {$\scriptstyle\mu NL$};
	
	\node[squarednode] at (2+1, 0) {RN$_M$};
	\node[squarednode] at (3.75+3, 0) {UE$_K$};
	
	\node (ch) at (2+1, -0.6)
	{\includegraphics[scale=.3]{content/cache.pdf}};
	\draw[<->] (2.3+1, -0.23-0.6) -- (2.3+1, 0.24-0.6) node[pos=0.5,sloped, below] {$\scriptstyle\mu NL$};
	
	\draw[->, black, thick] (0.6, 2.6+2) to [out=-10, in=-230] (3.25+3, 0);
	\node at (3.45, 3.6) {$g_{K}$};
	\draw[->, thick, blue] (0.6, 2.6+2) -- (1.4+1, 0) node[pos=0.5,sloped,above] {$f_M$};
	\draw[->, magenta, thick] (2.45+1, 0.6+2) -- (3.25+3, 0) node[pos=0.5,sloped,below] {$h_{K1}$};
	\draw[->, magenta, thick] (2.5+1, 0) -- (3.25+3, 2.5+2) node[pos=0.5,sloped,below] {$h_{1M}$};
	\draw[->, magenta, thick] (2.5+1, 0) -- (3.25+3, 0) node[pos=0.5,sloped,below] {$h_{KM}$};
	
	\path (3.75+3, 2.6+2) -- (3.75+3, 0) node [black, scale=3.0, midway, sloped] {$\dots$};
	\path (2+1, 0.24+2) -- (2+1, 0) node [black, scale=1.5, midway, sloped] {$\dots$};
	\end{tikzpicture}
}

\newcommand{\PlotNDTMOneKOne}[0]{
		\begin{axis}[
			axis x line=bottom, axis y line=left,
			ymin=0, ymax=3, ytick={1,...,2}, ylabel={\small Normalized Delivery Time},
			xmin=0, xmax=1.2, xtick={1,...,1.2}, xlabel=$\mu$, grid=major
			]
			
			\addplot+[name path=f1, no markers, decoration={
				text={}, raise=1ex,
				text align={center}
			}, postaction={decorate}, very thick, blue, domain=0:1, smooth]{2-x};
			\addplot[very thick, blue, domain=1:1.2, smooth]{1};
			
			\path[name path=axis] (axis cs:0,3) -- (axis cs:0,2) -- (axis cs:1,1) -- (axis cs:1.2,1) -- (axis cs:1.2,3) -- (axis cs:0,3);
			\addplot[blue!10] fill between[of=axis and f1];
			\addplot[red, mark=otimes*, mark size=3.75pt, only marks] coordinates {(0,2) (1,1)};
		\end{axis}
}

\newcommand{\PlotNDTMOneKTwo}[0]{
	\begin{axis}[
	axis x line=bottom, axis y line=left,
	ymin=0, ymax=4, ytick={1,3}, ylabel={\small Normalized Delivery Time},
	xmin=0, xmax=1.2, xtick={1}, xlabel=$\mu$, grid=major,
	]
	
	\addplot+[name path=f1, no markers, decoration={
		text={}, raise=1ex,
		text align={center}
	}, postaction={decorate}, very thick, blue, domain=0:1, smooth]{3-x-x};
	\addplot[very thick, blue, domain=1:1.2, smooth]{1};
	
	\path[name path=axis] (axis cs:0,4) -- (axis cs:0,3) -- (axis cs:1,1) -- (axis cs:1.2,1) -- (axis cs:1.2,4) -- (axis cs:0,4);
	\addplot[blue!10] fill between[of=axis and f1];
	\addplot[red, mark=otimes*, mark size=3.75pt, only marks] coordinates {(0,3) (1,1)};
	\end{axis}
}

\newcommand{\PlotNDTMOneKGreatTwo}[0]{
	\begin{axis}[
	axis x line=bottom, 
	axis y line=left,
	ymin=0, 
	ymax=5, 
	ytick={1.5,1.6,2.5, 4}, 
	ylabel={\small Normalized Delivery Time},
	xmin=0, 
	xmax=1.2, 
	xtick={0.5, 0.8,1}, 
	xlabel=$\mu$, 
	xticklabels={$\nicefrac{1}{2}$,$\mu_1=\frac{K+1}{2K-1}$,$1$},
	yticklabels={$\ $,$\ $,$\frac{K+2}{2}$, \scriptsize{$K+1$}}, 
	grid=major, 
	]
	
	\addplot+[name path=f1, no markers, decoration={
		text along path,
		text={}, raise=1ex,
		text align={center}
	}, postaction={decorate}, very thick, blue, domain=0:0.8, smooth]{4-x-x-x};
		\addplot+[no markers, decoration={
			text={}, raise=1ex,
			text align={center}
		}, postaction={decorate}, very thick, blue, domain=0.8:1, smooth] plot coordinates {
		(0.8,1.6)
		(1,1.5)
	};
	
	\addplot+[no markers, decoration={
				text={}, raise=1ex,
				text align={center}
			}, postaction={decorate}, red, dashed, domain=0:1, smooth] plot coordinates {
			(0,4)
			(1,1.5)
	};
		
	\addplot[blue, very thick, domain=1:1.2, smooth]{1.5};
	
	\path[name path=axis1] (axis cs:0,5) -- (axis cs:0,4) -- (axis cs:0.8,1.6) -- (axis cs:1,1.5) -- (axis cs:1.2,1.5) -- (axis cs:1.2,5) -- (axis cs:0,5);
	\addplot[blue!10] fill between[of=axis1 and f1];
	\addplot[red, mark=otimes*, mark size=3pt, only marks] coordinates {(0,4) (0.8, 1.6) (1,1.5)};
	\addplot[blue, mark=square*, mark size=3pt, only marks] coordinates {(0.5, 2.5)};
	\draw [dashed] (axis cs:0.8,1.6) -- (axis cs:0.8,5);
	\draw[->] (axis cs:0.8,4) -- node [text width=2.0cm,above] {$\:\:\:\text{DoF}_{\Sigma}<2$}(axis cs:0.5,4);
	\draw[->,red!100] (axis cs:0.8,4) -- node [text width=2.0cm,above] {$\:\:\:\:\text{DoF}_{\Sigma}=2$}(axis cs:1.1,4);	
	\end{axis}
	\node at (-.3,1.4) {$\frac{K}{2}$};
	\node at (-.5,2.0) {$\frac{K^2-1}{2K-1}$};
}

\newcommand{\FilesKThree}[0]{
	\centering
\begin{tabular}{c|ccccc}
\multicolumn{1}{c}{} & \multicolumn{4}{c}{\textcolor{red}{RN \& DeNB}} & \textcolor{blue}{DeNB only}\\   
File $W_1$ & \tikz[remember picture]{\node[anchor=base,inner sep=0pt] (a1) {\textcolor{OliveGreen}{$\eta_{1,1}$}};} & \textcolor{OliveGreen}{$\eta_{1,2}$} & \textcolor{OliveGreen}{$\eta_{1,3}$} & $\eta_{1,4}$ & \tikz[remember picture]{\node[anchor=base,inner sep=0pt] (a5) {$\eta_{1,5}$};} \\
File $W_2$ & \textcolor{OliveGreen}{$\eta_{2,1}$} & \textcolor{OliveGreen}{$\eta_{2,2}$} & \textcolor{OliveGreen}{$\eta_{2,3}$} & $\eta_{2,4}$ & $\eta_{2,5}$ \\
File $W_3$ & \textcolor{OliveGreen}{$\eta_{3,1}$} & \textcolor{OliveGreen}{$\eta_{3,2}$} & \textcolor{OliveGreen}{$\eta_{3,3}$} & $\eta_{3,4}$ & $\eta_{3,5}$ \\
File $W_4$ & $\eta_{4,1}$ & $\eta_{4,2}$ & $\eta_{4,3}$ & \tikz[remember picture]{\node[anchor=base,inner sep=0pt] (d4) {$\eta_{4,4}$};} & \tikz[remember picture]{\node[anchor=base,inner sep=0pt](d5) {$\eta_{4,5}$};}\\ [0.5\baselineskip] \multicolumn{1}{c}{} & \multicolumn{3}{c}{\textcolor{OliveGreen}{ZF symbols}}
\tikz[remember picture,overlay]{
\node [draw = red, fit = (a1) (d4)] {};
}
\tikz[remember picture,overlay]{
	\node [draw = blue, fit = (a5) (d5)] {};
}
\end{tabular}
}

\newcommand{\ZFKThree}[0]{
	\centering
	\begin{tabular}{c|ccc}
		 & \multicolumn{3}{c}{ZF map}\\ \hline 
		UE$_1$ &  \textcolor{OliveGreen}{$\eta_{2,1}$} & \textcolor{OliveGreen}{$\eta_{2,2}$} & \textcolor{OliveGreen}{$\eta_{3,3}$} \\ \hline 
		UE$_2$ & \textcolor{OliveGreen}{$\eta_{3,1}$} & \textcolor{OliveGreen}{$\eta_{3,2}$} & \textcolor{OliveGreen}{$\eta_{1,3}$} \\ \hline 
		UE$_3$ & \textcolor{OliveGreen}{$\eta_{1,1}$} & \textcolor{OliveGreen}{$\eta_{1,2}$} & \textcolor{OliveGreen}{$\eta_{2,3}$} \\ \hline
	\end{tabular}
}

\newcommand{\AlignGraph}[0]{
	
	\draw[fill=blue!20,draw=white] (-3.5,-0.5) rectangle (-0.5,5.5);
	\draw[fill=green!20,draw=white] (4.5,-0.5) rectangle (-0.5,5.5);
	\draw[fill=magenta!20,draw=white] (4.5,-0.5) rectangle (10,5.5);
	
	\draw [thick,decoration={brace,raise=0.1cm
	},decorate] (-3.5,5.5) -- (-0.5,5.5) 
	node [pos=0.5,anchor=north,yshift=1.15cm,align=left] {\small \textcolor{blue}{Layer 1}: Align 2 \\\small symbols per UE};
	\draw [thick,decoration={brace,raise=0.1cm
	},decorate] (-0.5,5.5) -- (4.5,5.5) 
	node [pos=0.5,anchor=north,yshift=1.15cm,align=left] {\small \textcolor{OliveGreen}{Layer 2}: Align 3 \\\small symbols per UE};
	\draw [thick,decoration={brace,raise=0.1cm
	},decorate] (4.5,5.5) -- (10,5.5) 
	node [pos=0.5,anchor=north,yshift=1.15cm,align=left] {\small \textcolor{magenta}{Layer 3}: Align 3 \\\small symbols per UE};
	
	\begin{scope}[every node/.style={circle,thick,draw}]
			\node (d5) at (-3,2.5) {\footnotesize $\eta_{4,5}$};
			\node (c4) at (-0.5,0) {\footnotesize $\eta_{3,4}$};
			\node (c2) at (2.0,0) {\footnotesize $\eta_{3,2}$};
			\node (b5) at (4.5,0) {\footnotesize $\eta_{2,5}$};
			\node (a3) at (7.0,0) {\footnotesize $\eta_{1,3}$};
			\node (b1) at (9.5,0) {\footnotesize $\eta_{2,1}$};		
			\node (b4) at (-0.5,2.5) {\footnotesize $\eta_{2,4}$};
			\node (b2) at (2.0,2.5) {\footnotesize $\eta_{2,2}$};
			\node (a5) at (4.5,2.5) {\footnotesize $\eta_{1,5}$};
			\node (c3) at (7.0,2.5) {\footnotesize $\eta_{3,3}$};
			\node (a1) at (9.5,2.5) {\footnotesize $\eta_{1,1}$};		
			\node (a4) at (-0.5,5) {\footnotesize $\eta_{1,4}$};
			\node (a2) at (2.0,5) {\footnotesize $\eta_{1,2}$};
			\node (c5) at (4.5,5) {\footnotesize $\eta_{3,5}$};
			\node (b3) at (7.0,5) {\footnotesize $\eta_{2,3}$};
			\node (c1) at (9.5,5) {\footnotesize $\eta_{3,1}$};
		\end{scope}
	
	\begin{scope}[>={Stealth[black]},
		every node/.style={fill=white,circle},
		every edge/.style={draw=red,very thick}]
		\path [<->] (d5) edge node {\footnotesize UE$_1$} (b4);
		\path [<->] (d5) edge node {\footnotesize  UE$_2$} (c4);
		\path [<->] (d5) edge node {\footnotesize  UE$_3$} (a4);
		
		\path [<-] (b4) edge node {\footnotesize  UE$_3$} (b2);
		\path [->] (b2) edge node {\footnotesize  UE$_3$} 
		(a5);
		
		\path [<-] (c4) edge node {\footnotesize  UE$_1$} (c2);
		\path [->] (c2) edge node {\footnotesize  UE$_1$}
		(b5);
		
		\path [<-] (a4) edge node {\footnotesize  UE$_2$} (a2);
		\path [->] (a2) edge node {\footnotesize  UE$_2$}
		(c5);
		
		\path [<-] (a5) edge node {\footnotesize  UE$_2$} (c3);
		\path [->] (c3) edge node {\footnotesize  UE$_2$} 
		(a1);
		
		\path [<-] (b5) edge node {\footnotesize  UE$_3$} (a3);
		\path [->] (a3) edge node {\footnotesize  UE$_3$} 
		(b1);
		
		\path [<-] (c5) edge node {\footnotesize  UE$_1$} (b3);
		\path [->] (b3) edge node {\footnotesize  UE$_1$} 
		(c1);

	\end{scope}
}

%% file: Caching_final_AC_2nd_submission.bbl
 \newcommand{\noop}[1]{}
\begin{thebibliography}{10}
\providecommand{\url}[1]{#1}
\csname url@samestyle\endcsname
\providecommand{\newblock}{\relax}
\providecommand{\bibinfo}[2]{#2}
\providecommand{\BIBentrySTDinterwordspacing}{\spaceskip=0pt\relax}
\providecommand{\BIBentryALTinterwordstretchfactor}{4}
\providecommand{\BIBentryALTinterwordspacing}{\spaceskip=\fontdimen2\font plus
\BIBentryALTinterwordstretchfactor\fontdimen3\font minus
  \fontdimen4\font\relax}
\providecommand{\BIBforeignlanguage}[2]{{%
\expandafter\ifx\csname l@#1\endcsname\relax
\typeout{** WARNING: IEEEtran.bst: No hyphenation pattern has been}%
\typeout{** loaded for the language `#1'. Using the pattern for}%
\typeout{** the default language instead.}%
\else
\language=\csname l@#1\endcsname
\fi
#2}}
\providecommand{\BIBdecl}{\relax}
\BIBdecl

\bibitem{network_m2_2011}
{Network, Evolved Universal Terrestrial Radio Access}, ``M2 {Application}
  {Protocol} ({M}2ap) ({Release} 10),'' 2011.

\bibitem{Kakar_Thesis}
J.~Kakar, ``{UAV Communications: Spectral Requirements, MAV and SUAV Channel
  Modeling, OFDM Waveform Parameters, Performance and Spectrum Management},''
  Master's thesis, Virginia Tech, Blacksburg, USA, 2015.

\bibitem{KakarUAV}
\BIBentryALTinterwordspacing
J.~Kakar and V.~Marojevic, ``{Waveform and Spectrum Management for Unmanned
  Aerial Systems Beyond 2025},'' 2017. [Online]. Available:
  \url{http://arxiv.org/abs/1708.01664}
\BIBentrySTDinterwordspacing

\bibitem{Maddah-Ali2}
M.~A. Maddah-Ali and U.~Niesen, ``Fundamental limits of caching,'' \emph{Trans.
  on Info. Theory}, vol.~60, no.~5, pp. 2856--2867, May 2014.

\bibitem{Soheil}
S.~Gherekhloo and A.~Sezgin, ``{Latency-Limited Broadcast Channel with
  Cache-Equipped Helpers},'' \emph{IEEE Trans. on Wireless Communications},
  vol.~16, no.~7, pp. 4192--4203, July 2017.

\bibitem{Maddah_Ali}
M.~A. Maddah-Ali and U.~Niesen, ``Cache-aided interference channels,'' in
  \emph{ISIT}, June 2015, pp. 809--813.

\bibitem{avik}
A.~Sengupta, R.~Tandon, and O.~Simeone, ``Cache aided wireless networks:
  Tradeoffs between storage and latency,'' in \emph{CISS}, March 2016, pp.
  320--325.

\bibitem{Tandon}
R.~Tandon and O.~Simeone, ``Cloud-aided wireless networks with edge caching:
  Fundamental latency trade-offs in fog radio access networks,'' in
  \emph{ISIT}, July 2016, pp. 2029--2033.

\bibitem{KakarArxiv}
\BIBentryALTinterwordspacing
J.~Kakar, S.~Gherekhloo, and A.~Sezgin, ``{Fundamental Limits on Delivery Time
  in Cloud- and Cache-Aided Heterogeneous Networks},'' 2017. [Online].
  Available: \url{http://arxiv.org/abs/1706.07627}
\BIBentrySTDinterwordspacing

\bibitem{KakarICC}
J.~Kakar, S.~Gherekhloo, Z.~H. Awan, and A.~Sezgin, ``{Fundamental limits on
  latency in cloud- and cache-aided HetNets},'' in \emph{ICC}, May 2017, pp.
  1--6.

\bibitem{Naderializadeh}
N.~Naderializadeh, M.~A. Maddah-Ali, and A.~S. Avestimehr, ``Fundamental limits
  of cache-aided interference management,'' in \emph{{ISIT}}, July 2016, pp.
  2044--2048.

\bibitem{conference214}
J.~Kakar, S.~Gherekhloo, and A.~Sezgin, ``{Fundamental limits on latency in
  transceiver cache-aided HetNets},'' in \emph{ISIT}, June 2017, pp.
  2955--2959.

\bibitem{Liu2011}
Y.~Liu and E.~Erkip, ``Completion time in broadcast channel and interference
  channel,'' in \emph{Allerton}, Sept 2011, pp. 1694--1701.

\bibitem{Cadambe09}
V.~R. Cadambe and S.~A. Jafar, ``{Interference Alignment and the Degrees of
  Freedom of Wireless X Networks},'' \emph{IEEE Trans. on Info. Theory},
  vol.~55, no.~9, pp. 3893--3908, Sept 2009.

\bibitem{Motahari14}
A.~S. Motahari, S.~Oveis-Gharan, M.~A. Maddah-Ali, and A.~K. Khandani, ``{Real
  Interference Alignment: Exploiting the Potential of Single Antenna
  Systems},'' \emph{IEEE Trans. on Info. Theory}, vol.~60, no.~8, pp.
  4799--4810, Aug 2014.

\bibitem{Maddah-Ali10}
M.~A. Maddah-Ali, ``{On the degrees of freedom of the compound MISO broadcast
  channels with finite states},'' in \emph{ISIT}, June 2010, pp. 2273--2277.

\bibitem{Frank2014}
D.~Frank, K.~Ochs, and A.~Sezgin, ``{A systematic approach for interference
  alignment in CSIT-less relay-aided X-networks},'' in \emph{WCNC}, April 2014,
  pp. 1126--1131.

\end{thebibliography}
